\documentclass[pdftex]{sajl1}

\newif\ifarxiv  
\arxivtrue   

\ifarxiv
  \volume{Preprint}
  \issue{arXiv v2}
\else
  \volume{X}
  \issue{X}
\fi
\year{2026}
\usepackage{amsmath,amssymb,amsfonts,mathtools}
\usepackage[T1]{fontenc}
\usepackage[utf8]{inputenc}
\usepackage{lmodern}
\usepackage[hidelinks]{hyperref}
\usepackage[expansion=false]{microtype}
\newtheorem{definition}{Definition}[section]
\newtheorem{theorem}[definition]{Theorem}

\newtheorem{proposition}[definition]{Proposition}

\newtheorem{example}[definition]{Example}

\newcommand{\negr}[1]{\boldsymbol{#1}}
\newenvironment{proof}{\noindent\textbf{Proof.}\ }{\hfill $\negr{\blacksquare}$\par}

\newcommand{\Expr}{\mathsf{Expr}}
\newcommand{\BlindExpr}{\mathsf{BlindExpr}}
\newcommand{\Token}{\mathsf{Token}}

\newcommand{\Rat}{\mathbb{Q}}

\newcommand{\Encl}{\mathsf{Encl}}
\newcommand{\BlindEncl}{\mathsf{BlindEncl}}
\newcommand{\Eval}{\mathsf{eval}}
\newcommand{\TokenConsistent}{\mathsf{TokenConsistent}}
\newcommand{\TokenEnv}{\mathsf{TokenEnv}}
\newcommand{\Meas}[3]{\mathsf{Meas}(#1,#2,#3)}
\newcommand{\Exact}[2]{\mathsf{Exact}(#1,#2)}
\newcommand{\forgetTokens}{\mathsf{forgetTokens}}
\newcommand{\RewritesTo}{\mathrel{\rightsquigarrow}}
\newcommand{\Interchangeable}{\mathrel{\rightleftharpoons}}
\newcommand{\OneWayOnly}{\mathrel{\mathsf{OneWayOnly}}}
\newcommand{\Licensed}{\mathsf{Licensed}}
\newcommand{\lean}[1]{\nolinkurl{#1}}

\title{Token-Sensitive Enclosure Semantics}{Token-Sensitive Enclosure Semantics for Measurement-Bearing Expressions}

\author{D. B. Hulak, A. F. Ramos, and R. J. G. B. de Queiroz}{David B. Hulak, Arthur F. Ramos, and Ruy J. G. B. de Queiroz}

\begin{document}

\maketitle

\begin{abstract}
Token identity is semantic information for measurement-bearing expressions.
Intervals, dimension tags, and token-erased syntax can say what values a
measured leaf may take, but they cannot say whether two occurrences name the
same observation or two fresh observations.  We give a small formal semantics
in which each measured leaf carries an interval of possible exact values and an
opaque observation-event token.  Here ``token'' means an identity for a
measurement event, not a lexical token of the source syntax.  The denotation of
an expression is its warranted enclosure: the set of exact values still
justified by hidden-value environments that assign one value to each
observation token and respect the declared intervals.  Over this semantics,
\(e \RewritesTo e'\) is a claim-tightening judgment, equivalently enclosure
containment \(\Encl(e') \subseteq \Encl(e)\), while interchangeability is
equality of enclosures.  The distinction is visible in cancellation, background
subtraction, and self-division: reusing one token gives interchangeability with
the expected simplified expression, while using distinct tokens gives only
one-way containment.  We prove that provenance-blind summaries of the kind
studied here, preserving intervals, dimension tags, and token-erased syntax,
are insufficient to recover the correct rewrite class.  The formal results are
mechanized in Lean 4 with no \texttt{sorry} or \texttt{admit} placeholders.
\end{abstract}


\section*{Introduction}

The same algebraic expression can have different logical content depending on how
its measured leaves were produced.  If a single measurement is written twice as
$x-x$, both occurrences denote the same hidden exact value and the result is
forced to be zero.  If two distinct measurements with the same interval
bound are written as $x_1-x_2$, the value zero is compatible with the
measurements, but it is not forced: the hidden values may differ.  Standard
mathematical notation tends to erase this difference.  So do many formal
models that retain interval bounds, dimensions, or expression shape while
forgetting the identity of the observation event.

A typical background-correction calculation has exactly this form.  Suppose a
sample reading \(s\) is corrected by subtracting a background reading \(b\).  If
the expression \((s+b)-b\) reuses the same recorded background observation, then
the background contribution cancels as a matter of denotation.  If the two
background leaves instead denote two freshly measured background readings with
the same advertised interval, the expression may differ from \(s\).  The
intervals and dimensions visible at the leaves are identical in the two cases;
only observation identity distinguishes them.

This paper isolates that distinction as a small problem in computational
logic.  We define a token-sensitive denotational semantics for a minimal
arithmetic language of measurement-bearing expressions.  A measured leaf
\(\Meas{t}{I}{d}\) carries an observation-event token \(t\), an interval \(I\)
of possible exact values, and a dimension tag \(d\).  The token is an opaque
name for the observation event, not a lexical token of the expression syntax.
A hidden-value environment assigns one rational value to each observation
token.  The warranted enclosure of an expression is the set of exact values
that remain possible under those assignments.  Repeated occurrences of the
same token are therefore evaluated with the same hidden value, while distinct
tokens may vary independently.

The second layer treats rewrite notation as enclosure logic rather than as an
ordinary program-reduction relation.  The judgment
\[
  e \RewritesTo e' \quad\text{means}\quad \Encl(e') \subseteq \Encl(e).
\]
Thus the target \(e'\) makes a stronger, more specific claim: every value it
warrants was already warranted by the source \(e\).  In logical terms, the
enclosure claim made by \(e'\) entails the weaker claim made by \(e\).  Equality
of enclosures is the stricter, bidirectional relation and is called
interchangeability.

The paper makes the following contributions.
\begin{enumerate}
\item We give a token-sensitive enclosure semantics for a minimal arithmetic
  language with exact constants and measured leaves.
\item We define one-way containment rewrites and enclosure equality
  interchangeability over that semantics.
\item We prove exact-fragment conservativity: without measured leaves, both
  rewrite classifications collapse to equality of exact denotations.
\item We mechanize three theorem families showing that cancellation, shared
  background subtraction, and self-division are interchangeable in the
  same-token case but only one-way in the distinct-token case.
\item We prove a provenance-blind limitation theorem: token-erased summaries
  can agree while the token-sensitive rewrite classes differ.
\end{enumerate}

The rest of the paper defines the semantics, introduces the rewrite
classification layer, proves exact-fragment conservativity, presents the three
classification families, proves the provenance-blind limitation result,
records the Lean artifact lineage, and situates the contribution relative to
rewriting, interval methods, measurement uncertainty, provenance, and
mechanized semantics.

\section{A token-sensitive expression semantics}

\subsection{Syntax}

Let \(\Token\) be an opaque set of observation-event identifiers and let
intervals be closed rational intervals \(I=[\ell,u]\) with \(\ell \leq u\).
Tokens are opaque in the formal sense that the semantics uses only equality of
tokens, not their internal structure.  The Lean development uses \(\Rat\) as
the value carrier.  This is a deliberate modeling choice rather than a
fundamental restriction: rationals are decidable, give total interval endpoint
comparisons without an analytic order theory, and let the mechanization stay
constructive without importing classical real analysis.  For the theorem
families proved here, the separating witnesses are rational interval endpoints,
so the same classifications would persist over \(\mathbb{R}\); a real-valued
development would mainly add analytic infrastructure.

\begin{definition}[Expressions]
Expressions are generated by
\[
\begin{array}{rcl}
e &::=& \Exact{q}{d}
    \mid \Meas{t}{I}{d}
    \mid e+e
    \mid e-e
    \mid e\cdot e
    \mid e/e
    \mid -e ,
\end{array}
\]
where \(q\in\Rat\), \(t\in\Token\), \(I\) is a rational interval, and \(d\) is
a dimension tag.
\end{definition}

The dimension component is present so that examples can keep track of the
intended physical kind of a measured quantity.  In the current mechanization,
however, the dimension type has a single syntactic constructor and the
semantics does not compute with dimensions.  The results below should therefore
not be read as a developed units-of-measure type system.

\subsection{Token environments and warranted enclosures}

\begin{definition}[Token environment]
A token environment is a function
\[
  \sigma : \Token \to \Rat
\]
assigning one hidden exact rational value to each observation token.
\end{definition}

\begin{definition}[Token consistency]
An environment \(\sigma\) is token-consistent with an expression \(e\), written
\(\TokenConsistent(\sigma,e)\), when every measured leaf
\(\Meas{t}{[\ell,u]}{d}\) occurring in \(e\) satisfies
\[
  \ell \leq \sigma(t) \leq u.
\]
The condition is indexed by tokens, not by syntactic positions.  Two
occurrences with the same token must therefore share the same hidden value.
Structurally, exact leaves impose no condition; a measured leaf imposes the
interval condition above; and each arithmetic constructor combines the
conditions of its immediate subexpressions by conjunction.
\end{definition}

Evaluation is defined by structural recursion on expressions.  For a token
environment \(\sigma\),
\[
\begin{aligned}
\Eval(\sigma,\Exact{q}{d}) &= q,\\
\Eval(\sigma,\Meas{t}{I}{d}) &= \sigma(t),\\
\Eval(\sigma,e_1\circ e_2) &=
  \Eval(\sigma,e_1)\circ\Eval(\sigma,e_2),\\
\Eval(\sigma,-e) &= -\Eval(\sigma,e),
\end{aligned}
\]
where \(\circ\in\{+,-,\cdot,/\}\) uses the corresponding rational operation.
Division is total in the mechanization: Lean~4 rational division satisfies
\(x/0=0\).  The self-division theorems below guard against this degenerate case
by requiring the interval to be strictly positive, so that token consistency
forces a nonzero denominator.

\begin{definition}[Warranted enclosure]
The token-sensitive warranted enclosure of \(e\) is
\[
  \Encl(e) =
  \{\,v\in\Rat \mid
    \exists \sigma.\ \TokenConsistent(\sigma,e)
       \wedge \Eval(\sigma,e)=v\,\}.
\]
\end{definition}

This definition is the central object of the paper.  It says that an
expression denotes exactly the rational values still justified by the declared
measurement intervals and by the equality constraints induced by token reuse.
This is the sense in which the enclosure is ``warranted'' by the available
measurement information.

Equivalently, token environments can be read as possible worlds.  A world fixes
one hidden exact value for each observation token, and token consistency selects
the worlds compatible with the interval declarations appearing in the
expression.  The enclosure \(\Encl(e)\) is then the image of \(e\)'s denotation
over all compatible worlds.  Within each compatible world, reusing a token
imposes an equality constraint between occurrences of that observation; using a
fresh token removes that constraint.

\begin{example}[Token identity changes an enclosure]
Let \(I=[2,5]\).  The expression \(\Meas{t}{I}{d}-\Meas{t}{I}{d}\) has
enclosure \(\{0\}\), because every consistent environment evaluates both
occurrences at the same value \(\sigma(t)\).  The expression
\(\Meas{t_1}{I}{d}-\Meas{t_2}{I}{d}\), with \(t_1\neq t_2\), can realize both
\(0\) and \(3\): choose \(\sigma(t_1)=\sigma(t_2)=2\) for \(0\), or
\(\sigma(t_1)=5,\sigma(t_2)=2\) for \(3\).  The visible intervals and
dimension tags are the same, but the token structure is not.
\end{example}

\begin{center}
\begin{tabular}{@{}ll@{}}
\textbf{Paper notation} & \textbf{Lean name}\\
\(\Token\) & \texttt{Token}\\
\(\Expr\) & \texttt{Expr}\\
\(\TokenEnv\) & \texttt{TokenEnv}\\
\(\TokenConsistent(\sigma,e)\) & \texttt{TokenConsistent sigma e}\\
\(\Encl(e)\) & \texttt{Enclosure e}\\
\(\Licensed(e,e')\) & \texttt{Licensed e e'}\\
\(e\RewritesTo e'\) & \texttt{RewritesTo e e'}\\
\(e\Interchangeable e'\) & \texttt{Interchangeable e e'}\\
\(\OneWayOnly(e,e')\) & \texttt{OneWayOnly e e'}\\
\(\forgetTokens\) & \texttt{forgetTokens}\\
\(\BlindEncl\) & \texttt{BlindEnclosure}
\end{tabular}
\end{center}

\section{Rewrite classification as enclosure logic}

The semantic objects in this section form a small order-theoretic logic of
measurement claims.  An expression denotes a set of possible exact values.
Containment is entailment between such claims: the smaller enclosure is the
stronger claim.  Interchangeability is observational equivalence at the level
of warranted values.  Token erasure is an abstraction from token-sensitive
expressions to provenance-blind expressions, and Section~5 proves an
insufficiency result for that abstraction: the token-sensitive rewrite class
is not determined by intervals, dimensions, and token-erased syntax alone.

\begin{definition}[Claim-tightening rewrite]
For expressions \(e,e'\), define
\[
  e \RewritesTo e'
  \quad\text{iff}\quad
  \Encl(e')\subseteq\Encl(e).
\]
The target expression makes a stronger, more specific enclosure claim: every
value it permits is already permitted by the source expression.  Thus
\(e\RewritesTo e'\) records semantic entailment from the target claim back to
the source claim.
\end{definition}

In the Lean source, the underlying denotational containment predicate is
\(\Licensed(e,e')\).  The paper writes the corresponding rewrite-classification
judgment as \(e\RewritesTo e'\), because the result being classified is a
rewrite claim.  The intended reading is claim tightening: the displayed target
is licensed because it entails the source enclosure claim.

\begin{definition}[Interchangeability and one-way-only rewrites]
The expressions \(e\) and \(e'\) are interchangeable, written
\(e\Interchangeable e'\), when
\[
  \Encl(e)=\Encl(e').
\]
The rewrite from \(e\) to \(e'\) is one-way-only when
\[
  e\RewritesTo e'
  \quad\text{and}\quad
  \neg(e'\RewritesTo e).
\]
\end{definition}

The terminology separates two logical readings of a rewrite.  If a rewrite is
intended to preserve all warranted values, interchangeability is the relevant
classification.  If the rewrite is intended as a one-way strengthening,
containment records that direction explicitly.

\begin{proposition}[Basic laws]
The relation \(\RewritesTo\) is reflexive and transitive.
Interchangeability is an equivalence relation.  Moreover,
\[
  e\Interchangeable e'
  \quad\text{iff}\quad
  (e\RewritesTo e')\ \text{and}\ (e'\RewritesTo e).
\]
\end{proposition}

\begin{proof}
All parts are immediate from set inclusion and set equality for enclosures.
The Lean theorem
\lean{interchangeable_iff_rewrites_both} is the mechanized form of the
last equivalence.
\end{proof}

\begin{proposition}[Exact-target criterion]
For any expression \(e\), rational \(q\), and dimension tag \(d\),
\[
  e\RewritesTo \Exact{q}{d}
  \quad\text{iff}\quad
  q\in\Encl(e).
\]
Also,
\[
  e\Interchangeable \Exact{q}{d}
  \quad\text{iff}\quad
  \Encl(e)=\{q\}.
\]
\end{proposition}

\begin{proof}
The enclosure of \(\Exact{q}{d}\) is the singleton \(\{q\}\).  The first
equivalence is therefore \(\{q\}\subseteq\Encl(e)\) iff \(q\in\Encl(e)\).
The second is the definition of interchangeability after reducing the exact
enclosure to that singleton.  These are mechanized as
\lean{rewrites_to_exact_iff_mem} and
\lean{interchangeable_exact_iff_singleton}.
\end{proof}

\section{Conservativity on exact expressions}

The token-sensitive semantics should not introduce distinctions in expressions
that have no measured leaves.

\begin{definition}[Exact fragment]
An expression is exact when every leaf is of the form \(\Exact{q}{d}\).
For such an expression \(e\), let \(\mathsf{exactValue}(e) = \Eval(\sigma, e)\)
for any token environment \(\sigma\); since every leaf is exact, the result is
independent of \(\sigma\) and equals the total rational evaluation under the
structural semantics of Section~1.2.
\end{definition}

\begin{theorem}[Exact-fragment conservativity]
If \(e\) and \(e'\) are exact expressions, then
\[
  e\RewritesTo e'
  \quad\text{iff}\quad
  \mathsf{exactValue}(e)=\mathsf{exactValue}(e')
\]
and
\[
  e\Interchangeable e'
  \quad\text{iff}\quad
  \mathsf{exactValue}(e)=\mathsf{exactValue}(e').
\]
Consequently there are no one-way-only rewrites between exact expressions in
this fragment.
\end{theorem}

\begin{proof}
An exact expression ignores the token environment, so every token-consistent
evaluation yields the same rational value.  Its enclosure is therefore the
singleton \(\{\mathsf{exactValue}(e)\}\).  Inclusion between two singleton
enclosures is equality of their elements, and equality of singleton enclosures
is the same equality.  The Lean development proves these statements as
\lean{rewrites_exact_iff_exactValue_eq},
\lean{interchangeable_exact_iff_exactValue_eq}, and
\lean{exact_expressions_not_one_way_only}.
\end{proof}

\section{Three token-sensitive rewrite families}

This section records the theorem families that demonstrate the difference
between token reuse and repeated measurement.

\subsection{Cancellation}

Fix a nondegenerate interval \(I=[\ell,u]\) with \(\ell<u\), a dimension tag
\(d\), and distinct tokens \(t_1\neq t_2\).

\begin{theorem}[Same-token cancellation]
For any token \(t\),
\[
  \Meas{t}{I}{d} - \Meas{t}{I}{d}
  \Interchangeable
  \Exact{0}{d}.
\]
\end{theorem}

\begin{proof}
Every token-consistent environment assigns one value \(\sigma(t)\) to both
occurrences, so the expression evaluates to
\(\sigma(t)-\sigma(t)=0\).  Conversely, a consistent environment exists by
choosing any endpoint of \(I\).  Thus the enclosure is \(\{0\}\).  In Lean this
is \lean{same_token_subtraction_interchangeable_exact_zero}.
\end{proof}

\begin{theorem}[Distinct-token cancellation]
For \(t_1\neq t_2\),
\[
  \Meas{t_1}{I}{d} - \Meas{t_2}{I}{d}
  \OneWayOnly
  \Exact{0}{d}.
\]
\end{theorem}

\begin{proof}
The value \(0\) is realizable by assigning both tokens the same endpoint of
\(I\), so the exact-target criterion gives the forward containment rewrite.
The reverse containment fails because the endpoint assignment
\(\sigma(t_1)=u\), \(\sigma(t_2)=\ell\) realizes \(u-\ell\), which is nonzero.
Thus the source enclosure is not the singleton \(\{0\}\).  The mechanized
statement is
\lean{distinct_token_subtraction_one_way_only_to_exact_zero}.
\end{proof}

\subsection{Background subtraction}

Let \(s=\Meas{t_s}{I_s}{d}\) be a signal measurement and let background
measurements use interval \(I_b\).  Assume the relevant tokens are distinct
where stated.

\begin{theorem}[Shared-background subtraction]
If the same background token \(t_b\) is used in both background positions and
\(t_s\neq t_b\), then
\[
\begin{aligned}
  &(\Meas{t_s}{I_s}{d}+\Meas{t_b}{I_b}{d})-\Meas{t_b}{I_b}{d}\\
  &\qquad\Interchangeable \Meas{t_s}{I_s}{d}.
\end{aligned}
\]
\end{theorem}

\begin{proof}
Both background occurrences evaluate to \(\sigma(t_b)\), so the background
contribution cancels exactly and the remaining values are precisely the signal
values.  The side condition \(t_s\neq t_b\) is genuine in the current
formalization: one token carrying two interval declarations would be governed
by their intersection.  The Lean theorem is
\lean{shared_offset_subtraction_interchangeable_signal}.
\end{proof}

\begin{theorem}[Distinct-background subtraction]
Let \(t_s\neq t_{b_1}\), \(t_s\neq t_{b_2}\), and
\(t_{b_1}\neq t_{b_2}\).  If \(I_b\) is nondegenerate, then
\[
\begin{aligned}
  &(\Meas{t_s}{I_s}{d}+\Meas{t_{b_1}}{I_b}{d})
     -\Meas{t_{b_2}}{I_b}{d}\\
  &\qquad\OneWayOnly\, \Meas{t_s}{I_s}{d}.
\end{aligned}
\]
\end{theorem}

\begin{proof}
The forward containment is witnessed by assigning the two background tokens the
same value, which realizes every signal value inside the larger expression's
enclosure.  The reverse containment fails because the nondegenerate background
interval provides endpoint assignments whose background difference is nonzero,
producing values outside the plain signal enclosure.  This is mechanized as
\lean{distinct_background_subtraction_one_way_only_to_signal}.
\end{proof}

\subsection{Self-division}

Fix a positive nondegenerate interval \(I=[\ell,u]\) with \(0<\ell<u\).

\begin{theorem}[Same-token self-division]
For any token \(t\),
\[
  \Meas{t}{I}{d}/\Meas{t}{I}{d}
  \Interchangeable
  \Exact{1}{d}.
\]
\end{theorem}

\begin{proof}
Token consistency and positivity imply \(\sigma(t)\neq 0\), and both
occurrences use that same nonzero value.  Therefore every evaluation yields
\(\sigma(t)/\sigma(t)=1\), and \(1\) is realized by any consistent assignment.
The Lean theorem is
\lean{same_token_division_interchangeable_exact_one_positive}.
\end{proof}

\begin{theorem}[Distinct-token self-division]
For \(t_1\neq t_2\),
\[
  \Meas{t_1}{I}{d}/\Meas{t_2}{I}{d}
  \OneWayOnly
  \Exact{1}{d}.
\]
\end{theorem}

\begin{proof}
The value \(1\) is realizable by assigning both tokens \(\ell\).  The reverse
containment fails because the endpoint assignment
\(\sigma(t_1)=u\), \(\sigma(t_2)=\ell\) realizes \(u/\ell\neq 1\).  The
mechanized statement is
\lean{distinct_token_division_one_way_only_to_exact_one_positive}.
\end{proof}

The expression \(\Exact{1}{d}\) follows the present syntactic convention that
exact leaves carry a dimension tag forward as written.  A richer dimension
algebra would normally assign a self-quotient a dimensionless type; that
refinement is outside the present formal core.

\section{A token-erasure limitation theorem for rewrite logic}

The preceding examples might suggest that the distinction can be recovered from
a richer token-erased summary.  The limitation theorem in this section rules out
that possibility for the token-erased summaries considered here, which preserve
all interval declarations, all dimension tags, and the token-erased expression
tree: such summaries still cannot determine the token-sensitive rewrite class.

\begin{definition}[Provenance-blind expression]
The type \(\BlindExpr\) mirrors the expression syntax except that measured
leaves carry only an interval and a dimension tag:
\[
  \mathsf{BlindMeas}(I,d).
\]
The forgetful map
\(\forgetTokens:\Expr\to\BlindExpr\) erases tokens from measured leaves and
acts homomorphically on arithmetic structure.
\[
\begin{aligned}
\forgetTokens(\Exact{q}{d}) &=& \mathsf{BlindExact}(q,d),\\
\forgetTokens(\Meas{t}{I}{d}) &=& \mathsf{BlindMeas}(I,d),\\
\forgetTokens(e_1\circ e_2) &=&
  \forgetTokens(e_1)\circ\forgetTokens(e_2),\\
\forgetTokens(-e) &=& -\forgetTokens(e),
\end{aligned}
\]
where \(\circ\in\{+,-,\cdot,/\}\).
\end{definition}

\begin{definition}[Blind enclosure]
The blind enclosure \(\BlindEncl\) interprets exact leaves as singletons,
measured leaves as their intervals, and arithmetic nodes compositionally.  It
does not have a shared-token environment, so repeated measured leaves are
treated through their interval summaries rather than through observation
identity.
\[
\begin{aligned}
\BlindEncl(\mathsf{BlindExact}(q,d)) &=& \{q\},\\
\BlindEncl(\mathsf{BlindMeas}([\ell,u],d)) &=&
  \{v\in\Rat\mid \ell\leq v\leq u\},\\
\BlindEncl(a+b) &=&
  \{v\mid \exists x\in\BlindEncl(a),\, y\in\BlindEncl(b).\\
  &&\hspace{5em} x+y=v\},
\end{aligned}
\]
with analogous compositional clauses for subtraction, multiplication,
division (using the same total rational division), and negation.
\end{definition}

\begin{theorem}[Blind agreement with rewrite-class disagreement]
Each of the following pairs has equal token-erased blind enclosure but different
token-sensitive rewrite class.
\begin{enumerate}
\item Same-token cancellation and distinct-token cancellation.
\item Shared-background subtraction and distinct-background subtraction.
\item Same-token self-division and distinct-token self-division.
\end{enumerate}
In every case, the first expression is interchangeable with the simplified
target, while the second is only one-way to that target.
\end{theorem}

\begin{proof}
For each pair, token erasure yields definitionally equal blind expressions, so
their blind enclosures are equal.  The rewrite-class separation is exactly the
classification theorem proved in the preceding section.  The corresponding
Lean theorem declarations are:
\begin{center}
\small
\begin{tabular}{@{}l@{}}
\lean{blind_cancellation_patterns_agree_but_rewrite_classes_differ}\\
\lean{blind_background_patterns_agree_but_rewrite_classes_differ}\\
\lean{blind_division_patterns_agree_but_rewrite_classes_differ}
\end{tabular}
\end{center}
\end{proof}

Thus token identity is not merely an implementation detail.  In this formal
fragment it is semantic information required to determine whether a candidate
simplification is an equality of warranted enclosures or only a one-way
entailment.  Intervals, dimensions, and token-erased syntax are jointly
insufficient for that classification.

\section{Mechanization and artifact lineage}

The formalization is available in the public repository
\cite{measurementProvenanceSemanticsRepo}.  The manuscript was checked against
commit
\[
\texttt{eede9350e14b9e0d9bb392d0b5b54873587f01d5}.
\]
At that commit, the Lean development builds under the repository-local
toolchain using
\[
  \texttt{./scripts/with-local-toolchain.sh lake build}.
\]
The source is also checked to contain no \texttt{sorry} or \texttt{admit}
placeholders by the repository's per-module scripts, which grep for both
terms before reporting success.

The principal source files are:
\begin{center}
\begin{tabular}{@{}p{0.39\linewidth}p{0.50\linewidth}@{}}
\lean{Token.lean} & opaque token carrier\\
\lean{Syntax.lean} & dimensions, intervals, and expressions\\
\lean{Semantics.lean} & environments, evaluation, enclosure, and containment\\
\lean{CoreTheorems.lean} & denotational cancellation and background results\\
\lean{BlindComparator.lean} & token-erasing comparator syntax and semantics\\
\lean{BlindComparatorTheorems.lean} & blind-comparator insufficiency result\\
\lean{RewriteClassification.lean} & rewrite classifications and theorem families above
\end{tabular}
\end{center}

The proof architecture is intentionally layered.  The rewrite layer is defined
over the denotational core rather than by modifying it.  The public artifact
also contains scripts for checking the cancellation, background-subtraction,
and blind-comparator theorem surfaces independently.

The source contains a small number of ground rational-arithmetic closures using
\texttt{native\_decide}, all in concrete interval examples such as
\([1,2]\).  The general theorem families stated above are independent of those
concrete convenience lemmas.  The artifact claim here is therefore precise:
the Lean files compile, the source contains no \texttt{sorry} or
\texttt{admit} placeholders, and the manuscript states the theorem families
checked by that build.

\section{Related work}

\paragraph{Term rewriting and equational reasoning.}
Standard term rewriting studies reduction relations, confluence, termination,
completion, and equational reasoning \cite{BaaderNipkow1998}.  Our contribution
isolates a semantic side condition that may be relevant before a rewrite
relation is used on expressions
containing measurement observations: two syntactically similar occurrences may
or may not denote the same hidden value.

\paragraph{Interval arithmetic and dependence.}
Classical interval arithmetic propagates enclosures through arithmetic
expressions \cite{Moore1966,MooreKearfottCloud2009,HickeyJuVanEmden2001}.
The dependency problem is that repeated occurrences of a variable may be
treated independently, causing over-approximation
\cite{Neumaier1990}.  Affine arithmetic tracks first-order correlations to
reduce such over-approximation \cite{deFigueiredoStolfi2004}.  The present
semantics is much smaller: it tracks only exact token identity, but it uses that
identity to classify rewrite claims rather than to optimize numerical range
bounds.

\paragraph{Measurement uncertainty.}
The Guide to the Expression of Uncertainty in Measurement presents the
standard metrological view that measurement results are accompanied by
uncertainty information \cite{JCGM2008}, and its Supplement~1 gives a Monte
Carlo procedure for propagating distributions through a measurement model
\cite{JCGM101_2008}.  Both treatments work with probability distributions or
coverage intervals on output quantities.  Our formalization takes a
deliberately weaker input: a closed interval of warranted exact values, with no
probability or covariance structure.  We do not model uncertainty propagation
in the JCGM sense; rather, we study one logical consequence of reusing or not
reusing a concrete observation, given that bounds are already declared.

\paragraph{Abstract interpretation.}
The set-valued enclosure map is naturally related to abstract interpretation:
value sets form an abstraction of concrete hidden-value evaluations, and
containment is the relevant order on claims
\cite{CousotCousot1977}.  We do not develop a general abstract interpreter or
fixpoint theory here.  The connection is used only to position the semantics:
token-sensitive environments refine the concrete states over which the
enclosure abstraction ranges.

\paragraph{Units of measure.}
Type systems for dimensions and units of measure, such as Kennedy's work on
relational parametricity and units \cite{Kennedy1997}, address dimensional
consistency.  Our current dimension component is intentionally weaker: it is a
tag carried by syntax, not a units calculus.  This limitation keeps the
token-identity phenomenon separate from dimensional typing.

\paragraph{Provenance.}
Database provenance distinguishes information about where data came from and
why query results exist \cite{BunemanKhannaTan2001}.  Provenance semirings give
an algebraic account of such annotations \cite{GreenKarvounarakisTannen2007},
the database-provenance literature has been surveyed by Cheney, Chiticariu, and
Tan \cite{CheneyChiticariuTan2009}, broader surveys cover database, workflow,
and scientific-pipeline provenance together
\cite{HerschelDiestelkaemperBenLamine2017}, and the W3C PROV data model
standardizes provenance interchange \cite{MoreauGrothCheneyLeboMiles2015}.  Our
tokens are far less expressive than these provenance frameworks: they record
only observation identity, not derivation, agency, or workflow structure.  The
blind comparator theorem shows that even this minimal provenance can be
logically relevant to rewrite classification.

\paragraph{Mechanized semantics.}
The development is mechanized in Lean 4, a theorem prover and programming
language designed for extensible interactive theorem proving
\cite{deMouraUllrich2021}.  In spirit, the paper follows the tradition of
mechanized metatheory: definitions are small enough to audit, and the stated
metatheorems are checked by a proof assistant.  The novelty here is the
particular measurement/provenance rewrite phenomenon, not a new proof-assistant
technique.

\section{Limitations}\label{sec:limitations}

We separate \emph{modeling choices}, which delimit the formal core but do not
weaken the stated theorems, from \emph{practical limitations}, which point to
results not provided by this paper.

\paragraph{Modeling choices.}
\begin{enumerate}
\item The value carrier is \(\Rat\) rather than \(\mathbb{R}\).  For the theorem
  families proved here, the separating witnesses are rational interval
  endpoints; a real-valued development would mainly add analytic infrastructure.
\item Intervals are closed rational intervals \([\ell,u]\).  More general
  enclosure shapes, such as open, unbounded, or disjoint intervals, would change
  the enclosure construction but not the role of token identity.
\item Dimensions are syntactic tags rather than a units calculus.  The
  token-identity phenomenon is logically independent of dimensional typing.
\item The language has only the elementary arithmetic constructors needed to
  state the three theorem families.
\end{enumerate}

\paragraph{Practical limitations.}
\begin{enumerate}
\item No decision procedure.  We do not provide an algorithm that decides
  \(e\RewritesTo e'\) or \(e\Interchangeable e'\) for arbitrary expressions,
  and we do not characterize the fragment on which such a decision procedure
  would be feasible.
\item No probability or covariance model.  Tokens carry observation identity,
  not a joint distribution.  The semantics therefore cannot represent
  partially correlated measurements, repeated measurements with known
  measurement-error variance, or the Monte Carlo propagation of Supplement~1 to
  the GUM \cite{JCGM101_2008}.
\item Atomic token identity only.  Two occurrences either share a token or do
  not.  Graded notions such as ``measured by the same instrument but on
  different days'' are not expressible in the present semantics; they would
  require equivalence relations or partial-equality structure on \(\Token\).
\item No completion or normalization theory.  The rewrite layer classifies
  individual rewrite claims; it does not give confluence, termination, or
  optimizer-soundness theorems.  Establishing such properties would require a
  strategy for choosing rewrite directions and a termination argument.
\item No empirical study.  The three theorem families are mechanized examples,
  not measurements of how often the distinction arises in scientific software.
  A corpus study of measurement-bearing code would be needed to assess
  practical frequency.
\end{enumerate}

These are explicit limits on the claim, not hidden assumptions.  Within the
small mechanized semantics, token identity can be necessary to determine whether
a candidate simplification is interchangeable with its source or only one-way by
enclosure containment.

\section{Conclusion}

Measurement-bearing expressions need a way to distinguish reusing one
observation from making a second observation with the same visible bounds.  A
token-sensitive enclosure semantics supplies that distinction with very little
machinery: token-consistent environments assign one hidden exact value per
observation token, and warranted enclosures collect the possible results.  Over
that semantics, equality of enclosures and one-way containment become different
rewrite classifications.

The Lean development proves that the distinction is not cosmetic.  Same-token
cancellation, shared-background subtraction, and same-token self-division are
interchangeable with their simplified forms.  Their distinct-token lookalikes
are only one-way.  A provenance-blind summary that preserves intervals,
dimension tags, and token-erased syntax cannot recover this classification.
The exact fragment remains conservative: without measurements, the rewrite
classifications collapse to ordinary equality.

The three examples also show that the split is not a peculiarity of
subtraction.  It appears around zero, around a scientifically common
background-subtraction pattern, and around a positive self-quotient.  The
operation changes; the semantic reason does not.

Several precise open problems follow from the present formal core.
\begin{enumerate}
\item \emph{Decidability of the rewrite classifications.}  Identify fragments
  of the expression language on which \(e\RewritesTo e'\) and
  \(e\Interchangeable e'\) are decidable.  The affine fragment, with sums,
  differences, and multiplication by exact rationals, and the linear-over-token
  fragment look most tractable, since their warranted enclosures are induced by
  linear constraints over the token environment.
\item \emph{Equational theory of enclosure equivalence.}  Characterize, by a
  complete set of axioms or rewrite rules, the largest class of algebraic
  identities that remain sound as enclosure interchangeability for
  measurement-bearing expressions.  Exact-fragment conservativity is the base
  case; the open question is which token-sensitive identities extend it.
\item \emph{Probabilistic and correlated extensions.}  Replace token
  environments by probability distributions over hidden-value assignments, and
  augment token identity with covariance or correlation structure, so that the
  framework can express partially correlated measurements and connect to the
  JCGM~101 Monte Carlo view \cite{JCGM101_2008}.  A first step is to identify
  which results of this paper survive when ``set of warranted values'' is
  replaced by ``warranted distribution over output values.''
\end{enumerate}

Solving any of these problems likely requires a richer rewrite calculus than
the one developed here, and probably a closer connection to abstract
interpretation or to relational interval domains.

The broader lesson for formal systems that manipulate measurement-bearing
expressions is simple: interval bounds and syntax alone may not carry enough
logical information.  Observation identity can be part of the semantics of a
rewrite claim.

\authorname{David B. Hulak}
\address{Independent Researcher}
\email{dbhulak@gmail.com}

\authorname{Arthur F. Ramos}
\address{Microsoft}
\email{arfreita@microsoft.com}

\authorname{Ruy J. G. B. de Queiroz}
\address{Centro de Inform\'atica, Universidade Federal de Pernambuco}
\email{ruy@cin.ufpe.br}

\end{document}